\documentclass[journal,10pt,onecolumn,draftclsnofoot,]{IEEEtran}
\usepackage[utf8]{inputenc}
\usepackage{enumerate}
\usepackage{multirow,array}
\usepackage{cite}
\usepackage{graphicx}
\usepackage{psfrag}
\usepackage{caption}
\usepackage{subcaption}
\usepackage{url}
\usepackage{amsmath}
\usepackage{bm}
\usepackage{amsthm}
\usepackage{array}
\usepackage{amssymb}
\usepackage{amsfonts}
\usepackage{float}
\usepackage{tabu}
\usepackage{algorithm}
\usepackage{algorithmicx}
\usepackage{algcompatible}
\usepackage{algpseudocode}
\usepackage{hyperref}
\usepackage{color, colortbl}
\usepackage{xcolor}
\usepackage{physics}

\hypersetup{
    colorlinks,
    citecolor=red,
    filecolor=black,
    linkcolor=blue,
    urlcolor=black
}

\newtheorem*{conjecture*}{Conjecture}
\newtheorem{lemma}{Lemma}

\newtheorem{remark}{Remark}
\newtheorem{definition}{Definition}
\newtheorem{theorem}{Theorem}

\newcommand{\hE}{\hat{E}}

\begin{document}

\title{On the Efficacy of the Peeling Decoder for the Quantum Expander Code}
 \author{%
 \IEEEauthorblockN{Jefrin Sharmitha Prabhu$^{\dagger\star}$, Abhinav Vaishya$^{\dagger\star}$, Shobhit Bhatnagar$^\dagger$, 	Aryaman Manish Kolhe$^\wedge$,\\
 V. Lalitha$^\#$ and P. Vijay Kumar$^\dagger$}\\
\IEEEauthorblockA{
$^\dagger$Department of Electrical Communication Engineering, IISc Bangalore\\
$^\wedge$Centre for Quantum Science and Technology, IIIT Hyderabad\\
$^\#$Signal Processing and Communication Research Center, IIIT Hyderabad\\
\{jefrinsharmithaprabhu,vaishyaabhinav,shobhitb97,aryaman.mk613,pvk1729\}@gmail.com, lalitha.v@iiit.ac.in\\
\thanks{$^\star$ indicates equal contribution.} 
}}
\maketitle

\begin{abstract}
The problem of recovering from qubit erasures has recently gained attention as erasures occur in many physical systems such as photonic systems, trapped ions, superconducting qubits and circuit quantum electrodynamics.  While several linear-time decoders for error correction are known, their error-correcting capability is limited to half the minimum distance of the code, whereas erasure correction allows one to go beyond this limit.  As in the classical case, stopping sets pose a major challenge in designing efficient erasure decoders for quantum LDPC codes.  In this paper, we show through simulation, that an attractive alternative here, is the use of quantum expander codes in conjunction with the peeling decoder that has linear complexity.  We also discuss additional techniques including small-set-flip decoding, that can be applied following the peeling operation, to improve decoding performance and their associated complexity. 

\end{abstract}

\section{Introduction}
Quantum information processing promises an exponential speed-up for many classical tasks. However, physical qubits are far more susceptible to noise than classical bits. Quantum error correction is imperative for building a practical quantum computer, and has been an active field of research for the last two decades.

A popular class of quantum error-correcting codes is the class of stabilizer codes developed by Gottesman \cite{gottesman1997stabilizer}, where the codespace is defined as the simultaneous $+1$-eigenspace of an abelian group of Pauli operators, called the stabilizer group. A practical constraint on the stabilizer operators is that they must have low weight, i.e., they must operate on only a few qubits at a time. Further, each qubit must be operated upon by only a few stabilizer operators. This gives rise to the notion of quantum Low Density Parity Check (LDPC) codes, analogous to the classical case.
Gottesman \cite{gottesman2013fault} showed that it is possible to build fault-tolerant quantum circuits using such quantum LDPC codes.
Efforts to construct good quantum LDPC codes that have both high rate and large minimum distances gave rise to many interesting and innovative ideas \cite{hastings2021fiber,breuckmann2021balanced,panteleev2021quantum,panteleev2022asymptotically,leverrier2022quantum,leverrier2023decoding,dinur2023good,gu2023efficient}. One such idea is the hypergraph product (HGP) construction by Tillich and Zémor \cite{tillich2013quantum}. The HGP construction takes two classical LDPC codes as inputs and yields a quantum LDPC code. A special case is when the classical codes are expander codes \cite{sipser1996expander}. In this case the resultant quantum code is called a quantum expander code \cite{leverrier2015quantum}.

While quantum expander codes do not simultaneously achieve a constant rate and linear minimum distance (unlike \cite{panteleev2022asymptotically,leverrier2022quantum,dinur2023good}), they are relevant to practice as they permit linear-time decoding \cite{leverrier2015quantum,fawzi2018efficient,fawzi2020constant}.

However, in addition to having good rates and distances, practical quantum error-correcting codes also need to have low decoding complexity. Recently, the problem of recovering from qubit erasures has gained much interest as they occur in many physical systems such as photonic systems \cite{knill2001scheme},\cite{bartolucci2023fusion} (where photon loss, which can be modeled as an erasure, is the dominant source of error) and systems based on neutral atoms \cite{wu2022erasure}, trapped ions \cite{kang2023quantum}, superconducting qubits \cite{kubica2023erasure} and circuit quantum electrodyanamics \cite{tsunoda2023error}.
Moreover, for the class of HGP codes, it has been shown using a classical Viderman-like approach \cite{viderman2013linear} that the problem of error correction can be converted to that of erasure correction for a certain error regime \cite{krishna2024viderman}. Several linear-time decoders for error-correction are known \cite{dinur2023good,gu2023efficient,leverrier2023decoding}, however their error-correcting capability is limited to half of the minimum distance of the code. Erasure correction allows one to go beyond this limit.
As in the classical case, stopping sets pose a major challenge in designing efficient erasure decoders for quantum LDPC codes.
Several decoders for erasure-correcting quantum LDPC codes have been proposed in the literature \cite{delfosse2020linear,delfosse2021almost,delfosse2022toward,lee2020trimming,solanki2023decoding,solanki2021correcting,yao2024cluster,kuo2024degenerate,gokduman2024erasure,freire2025optimizing} that deal with stopping sets in different ways. 

\emph{Our contributions:} The principal contribution of the paper is showing through simulation results that an attractive option for dealing with erasures, is the use of quantum expander codes in conjunction with the linear-complexity peeling decoder. We discuss additional techniques such as cluster-based decoding and small-set-flip decoding, that are capable of improving decoding performance by handling specific classes of stopping sets.  We also discuss their associated complexity. Our simulation results also include limited statistics on erasure patterns that the peeling decoder was unable to recover from. 


\emph{Organization of the paper:} Section~\ref{sec:background} provides background on quantum error correction. The HGP code, along with the relevant prior literature on decoding it, is presented in Section~\ref{sec:hgp}. Our peeling-based decoding algorithm is presented in Section~\ref{sec:main_result}. Analysis of our algorithm is presented in Section~\ref{sec:main_analysis}. Simulation results are presented in Section~\ref{sec:simulation} and the final section, Section~\ref{sec:conclusions}, draws conclusions. 

\section{Quantum Error Correction} \label{sec:background}

We will use $\mathbf{0}_n$ to denote the zero vector of length $n$, and $I_n$ to denote the $(n\times n)$ identity matrix.

\subsection{Classical Error-Correcting Codes}
\label{subsec:ecc} 
An \([n, k, d_{\mathrm{min}}]\) binary, linear, classical error-correcting code \(\mathcal{C}\) is a \(k\)-dimensional subspace of \(\mathbb{F}_2^n\), that corresponds to the nullspace of a suitably-defined matrix $H$ called the parity-check matrix of the code. 
 
The code can also be defined using a bipartite graph $\mathcal{T}(V \cup C, E)$, alternatively denoted as $\mathcal{T}(H)$, called the Tanner graph. The left vertices $V$ (variable nodes) and the right vertices $C$ (check nodes) correspond to the columns and rows of the parity-check matrix $H$ of size $|C| \times |V|$. An edge exists between $i \in C$ and $j \in V$ if $h_{ij} = 1$; otherwise, $h_{ij} = 0$.


A code is $(d_V, d_C)$-regular LDPC if its associated Tanner graph is biregular, with variable nodes $V$ having degree $d_V$ and check nodes $C$ having degree $d_C$. The particular class of biregular, bipartite graphs that we will consider here is the class of expander graphs \cite{sipser1996expander} defined below:
\begin{definition} \label{def:expander} (Expander graph)
    Let $\mathcal{T}(V \cup C, E)$ be a bipartite Tanner graph such that $|V| \geq |C|$. The graph $\mathcal{T}$ is said to be $(\gamma_V, \delta_V)$-left-expanding (resp. $(\gamma_C, \delta_C)$-right-expanding) for some constants $\gamma_V, \delta_V > 0$ (resp. $\gamma_C, \delta_C > 0$), if for any subset $S \subseteq V$ (resp. $\mathcal{T} \subseteq C$), such that $|S| \leq \gamma_V |V|$ (resp.$|T| \leq \gamma_C |C|$), the neighbourhood $\Gamma(S)$ of $S$ (resp. $\Gamma(T)$ of $\mathcal{T}$) in the graph $\mathcal{T}(V \cup C, E)$ satisfies $|\Gamma(S)| \geq (1-\delta_V) d_V |S|$ (resp. $|\Gamma(T)| \geq (1-\delta_C) d_C |T|$). The graph $\mathcal{T}(V \cup C, E)$ is said to be $(\gamma_V, \delta_V, \gamma_C, \delta_C)$-left-right-expanding if it is both $(\gamma_V, \delta_V)$-left-expanding and $(\gamma_C, \delta_C)$-right-expanding.
\end{definition}
The expander graphs considered in this paper are left-right-expanding.  Classical expander codes, introduced by Sipser and Spielman \cite{sipser1996expander}, are defined as codes associated with $(\gamma,\delta)$-left-expanding $(d_V,d_C)$-regular Tanner graphs $\mathcal{T}(V \cup C,E)$. The minimum distance $d$ of the expander code associated with Tanner graph $\mathcal{T}(V \cup C,E)$ satisfies the lower bound $d > \gamma |V|$. Viderman proposed in \cite{viderman2013linear}, a linear-time erasure decoding algorithm capable of correcting up to $\gamma |V|$ erasures.

Even though explicit construction of good expander graphs is challenging, such graphs can nevertheless be efficiently found using probabilistic techniques \cite{richardson2001efficient},\cite{fawzi2018efficient} as outlined in the theorem below.

\begin{theorem} (\cite[Theorem 2.3]{fawzi2018efficient}) 
Let $\delta_V,\delta_C > 0$ be two constants. For integers $d_V> \frac{1}{\delta_V}$ and $d_C> \frac{1}{\delta_C}$, a graph $\mathcal{T}(V \cup C, E)$ with left-degree bounded by $d_V$ and right-degree bounded by $d_C$ chosen at random according to some distribution is $(\gamma_V,\delta_V, \gamma_C, \delta_C)$-left-right expanding for $\gamma_V,\gamma_C=\Omega(1)$ with high probability.
\end{theorem}

\subsection{Quantum Error-Correcting Codes}
\label{subsec:qecc}
A quantum error-correcting code encoding $k$ logical qubits into $n$ physical qubits is a $2^k$-dimensional subspace of $(\mathbb{C}^{2})^{\otimes n}$. Here $\otimes$ denotes tensor product.
Gottesman developed the stabilizer formalism for quantum codes in \cite{gottesman1997stabilizer}, inspired by classical linear codes. The stabilizer formalism serves as a general framework for studying a large class of quantum code. We describe this formalism below.

The Pauli matrices are defined as
{\small
\[
    I_2=\begin{pmatrix}
        1 & 0\\
        0 & 1\\
    \end{pmatrix},~X=\begin{pmatrix}
        0 & 1\\
        1 & 0\\
    \end{pmatrix},~Z=\begin{pmatrix}
        1 & 0\\
        0 & -1\\
    \end{pmatrix},
\]
}
and $Y=iXZ$.
The Pauli group on $n$ qubits, $\mathcal{P}_n$, is the set of all Pauli operators acting on $n$ qubits, with usual matrix multiplication as the group operation:

$ \mathcal{P}_n=\{\omega P_1 \otimes P_2 \otimes \cdots \otimes P_n :P_i \in \{I_2, X, Z, Y\}~\forall~i,\omega=\{\pm{1},\pm{i}\}\}.$

The operator \( P = X^{a_1}Z^{b_1} \otimes X^{a_2}Z^{b_2} \otimes \cdots \otimes X^{a_n}Z^{b_n} \) can be represented by the vector 
\[
p=(a_1~a_2 \cdots a_n \mid b_1~ b_2 \cdots b_n) \in  \mathbb{F}_2^{2n}.
\]
This representation of an operator as a binary vector is called the \textit{symplectic representation}.

A stabilizer group, $\mathcal{S}$, is an abelian subgroup of the Pauli group such that $-I_{2^n} \notin \mathcal{S}$. A quantum stabilizer code $Q_{\mathcal{S}}$ is defined as the simultaneous $+1$ eigenspace of all the operators in $\mathcal{S}$, i.e.,
\[
    Q_{\mathcal{S}} = \{\ket{\psi} : S\ket{\psi} = \ket{\psi} ~ \forall ~ S \in \mathcal{S}\},
\]
where we have used Dirac's bra-ket notation.
The weight of a Pauli operator $P = \omega^l P_1 \otimes P_2 \otimes \cdots \otimes P_n$ is the number of indices $j \in [1,n]$ such that $P_j \neq I_2$. For a quantum stabilizer code $Q_{\mathcal{S}}$ defined via a stabilizer group $\mathcal{S}$, the minimum distance, $d_{\mathrm{min}}$  is defined as the minimum weight of all Pauli operators in $N(\mathcal{S}) \setminus \mathcal{S}$, where $N(\mathcal{S})$ is the normalizer of $\mathcal{S}$ in $\mathcal{P}_n$. This stabilizer code $Q_{\mathcal{S}}$ is denoted as a $[[n,k,d_{\mathrm{min}}]]$ code.

\subsection{{Quantum Erasure Channel}}
\label{subsec:erasure_channel}
The quantum erasure channel is a noise model in which each qubit is lost or erased independently with some probability. This loss can be detected, and the density matrix of the qubit, $\rho$, is replaced with a maximally mixed state, $\frac{I}{2}$. The maximally mixed state can be expressed as $\frac{I}{2} = \frac{1}{4}(\rho + X\rho X^\dagger + Y\rho Y^\dagger + Z\rho Z^\dagger)$ , implying that each erased qubit can be interpreted as being acted upon by the Pauli operators $I$, $X$, $Y$, and $Z$ with equal probability. In this manner, erasure correction is effectively converted into correction of errors with known locations, denoted by a vector $\varepsilon$. Replacing each qubit in the support of \(\varepsilon\) with a maximally mixed state is equivalent to the encoded state being subjected to a uniform Pauli error \(E\) supported on the erasure locations, i.e., \(\mathrm{supp}(E) \subseteq \mathrm{supp}(\varepsilon)\). 


\subsection{CSS Codes}
\label{subsec:css}
Calderbank-Shor-Steane (CSS) codes are a sub-class of stabilizer codes where each stabilizer generator is a tensor product of either $X$-type or $Z$-type Pauli matrices \cite{calderbank1996good,steane1996error}, and are one of the most well-studied classes of quantum codes. A CSS code is defined by two binary classical codes, $\mathcal{C}_X = \ker(H_X)$ and $\mathcal{C}_Z = \ker(H_Z)$, such that ${\mathcal{C}_{Z}^{\perp}} \subseteq \mathcal{C}_X$ (or equivalently, ${\mathcal{C}_{X}^{\perp}} \subseteq \mathcal{C}_Z$). Here, $H_X$ of size $r_X \times n$ and $H_Z$ of size $r_Z \times n$ represent the parity-check matrices of $\mathcal{C}_X$ and $\mathcal{C}_Z$, respectively. The basis of dual code ${\mathcal{C}_X}^\perp$ (resp. ${\mathcal{C}_Z}^\perp$) corresponds to the $X$-type (resp. $Z$-type) stabilizer generators and is used to correct $Z$-type (resp. $X$-type) errors.

The quantum code \( \mathrm{CSS}(\mathcal{C}_X, \mathcal{C}_Z) \) is a stabilizer code with stabilizer generators determined by the following check matrix:

\begin{equation*}
H = \begin{pmatrix}
H_X & 0 \\
0 & H_Z
\end{pmatrix}  
\end{equation*}

The dimension of the CSS code is given by $k=\text{dim}(\mathcal{C}_X/{\mathcal{C}_{Z}^\perp)}=\text{dim}(\mathcal{C}_Z/{\mathcal{C}_{X}^\perp)}=\text{dim}(\mathcal{C}_X) +\text{dim}( \mathcal{C}_Z)-n$. The minimum distance of CSS code is given by, $d_{\mathrm{min}}= \text{min}\{d_X,d_Z\}$, where, $d_X:=\text{min}\{w_H(E) \mid E \in \mathcal{C}_X \backslash {\mathcal{C}_{Z}^{\perp}}\}$ and  $d_Z:=\text{min}\{w_H(E) \mid E \in \mathcal{C}_Z\backslash {\mathcal{C}_{X}^\perp\}}$, where $w_H(\cdot)$ denotes Hamming weight. The CSS code \(\mathrm{CSS}(\mathcal{C}_X, \mathcal{C}_Z)\) is an $[[n,k,d_{\mathrm{min}}]]$  code.

The error operator $E$, is symplectically represented by the binary $2n$-tuple $(E_X, E_Z)$. The syndrome corresponding to $(E_X, E_Z)$ is given by $\sigma=(\sigma_{X}, \sigma_{Z})$, where $\sigma_{Z}= H_{X}  {E_{Z}^{T}}$  (resp. $\sigma_{X}= H_{Z} {E_{X}^{T}}$) is obtained by measuring $X$-type (resp. $Z$-type) stabilizer generators and is used for correcting $Z$-type (resp. $X$-type) errors independently.

Error correction for CSS codes can be performed independently for \(X\)-errors and \(Z\)-errors. Specifically, by considering the syndrome \(\sigma = (\sigma_X, \mathbf{0}_{r_z})\) (respectively, \(\sigma = (\mathbf{0}_{r_x}, \sigma_Z)\)), \(X\)-errors (respectively, \(Z\)-errors) can be corrected separately. In this paper, the description of the algorithm presented in Section \ref{sec:main_result} is restricted to only \(X\)-type errors \(E\), using the syndrome \(\sigma\) obtained by measuring the \(Z\)-type stabilizer generators.

Quantum LDPC codes are a class of stabilizer codes where each stabilizer generator acts on a constant number of qubits, and each qubit is acted upon by a constant number of stabilizer generators. When CSS codes are constructed with the classical parity-check matrices $H_X$ and $H_Z$ being sparse, they correspond to quantum LDPC codes. Precisely, these matrices have a constant row weight and a constant column weight, which characterizes the LDPC property.

\section{Hypergraph Product Codes}
\label{sec:hgp}
The HGP construction by Tillich and Zémor \cite{tillich2013quantum} provides a framework for constructing quantum CSS codes from two classical codes.

Let \(\mathcal{C}\) be a classical code with parity-check matrix \(H\). The transpose code \(\mathcal{C}^T\) is obtained as the null-space of \(H^T\). Informally, \(\mathcal{C}^T\) is the code obtained by reversing the roles of check and variable nodes in the Tanner graph of $\mathcal{C}$. 

Let \(\mathcal{C}_1\) be an \([n_1, k_1, d_1]\) classical code with parity-check matrix \(H_1\) of size $r_1 \times n_1$, and \(\mathcal{C}_2\) be an \([n_2, k_2, d_2]\) classical code with parity-check matrix \(H_2\) of size $r_2 \times n_2$.

The product code \(\mathcal{C}_1 \otimes \mathcal{C}_2\) is the classical code with a parity-check matrix defined as:
\[
H = \begin{pmatrix}
    I_{n_1} \otimes H_2 \\
    H_1 \otimes I_{n_2}
\end{pmatrix}
\]

The hypergraph product of two classical codes \(\mathcal{C}_1\) and \(\mathcal{C}_2\), denoted as \(\mathrm{HGP}(\mathcal{C}_1, \mathcal{C}_2)\), is defined as the quantum CSS code, \(\mathrm{CSS}(\mathcal{C}_X, \mathcal{C}_Z)\), where,  $\mathcal{C}_X^T = \mathcal{C}_1 \otimes \mathcal{C}_2^T$, and $\mathcal{C}_Z^T = \mathcal{C}_1^T \otimes \mathcal{C}_2$.
The check matrix of $\mathrm{HGP}(\mathcal{C}_1, \mathcal{C}_2)$ is given by \begin{equation*}
    H=\begin{pmatrix}
    H_X & 0 \\
    0 & H_Z \\
\end{pmatrix}
\end{equation*}
where $H_X = \begin{pmatrix}
    I_{n_1} \otimes H_2 & H_1^T \otimes I_{r_2}
\end{pmatrix}$ and $H_Z = \begin{pmatrix}
    H_1 \otimes I_{n_2} & I_{r_1} \otimes H_2^T
\end{pmatrix}$. It can be verified that \(H_X H_Z^T = 0\), satisfying the CSS property. Furthermore, if \(\mathcal{C}_1\) and \(\mathcal{C}_2\) are classical LDPC codes, then \(\mathrm{HGP}(\mathcal{C}_1, \mathcal{C}_2)\) is a quantum LDPC code.

Let $\mathcal{T}(H_1):=\mathcal{T}(V_1\cup C_1, E_1)$ be the Tanner graph for the classical code $\mathcal{C}_1$ and $\mathcal{T}(H_2):=\mathcal{T}(V_2\cup C_2, E_2)$  be the Tanner graph for the classical code $\mathcal{C}_2$. The Tanner graph of $\mathrm{HGP}(\mathcal{C}_1, \mathcal{C}_2)$ is a 4-partite graph $\mathcal{T}((V_1 \times V_2) \cup (C_1 \times C_2) \cup (V_1 \times C_2) \cup (C_1 \times V_2), E)$.

Here, the qubits are indexed by the set \(V_Q:=(V_1 \times V_2) \cup (C_1 \times C_2)\). The \(X\)-stabilizer generators are indexed by the set \(C_X:=C_1 \times V_2\), and the \(Z\)-stabilizer generators are indexed by the set \(C_Z:=V_1 \times C_2\). 

Two nodes \(u_1,u_2 \in V_Q \times (C_X \cup C_Z)\) are connected if: the first coordinate of \(u_1\) and \(u_2\) is the same, and the second coordinate is connected in \(\mathcal{T}(H_2)\), or
the second coordinate of \(u_1\) and \(u_2\) is the same, and the first coordinate is connected in \(\mathcal{T}(H_1)\). Fig. \ref{fig:HGP} illustrates this construction.

\begin{figure}
    \centering
   \scalebox{0.7} {\includegraphics[width=1\linewidth]{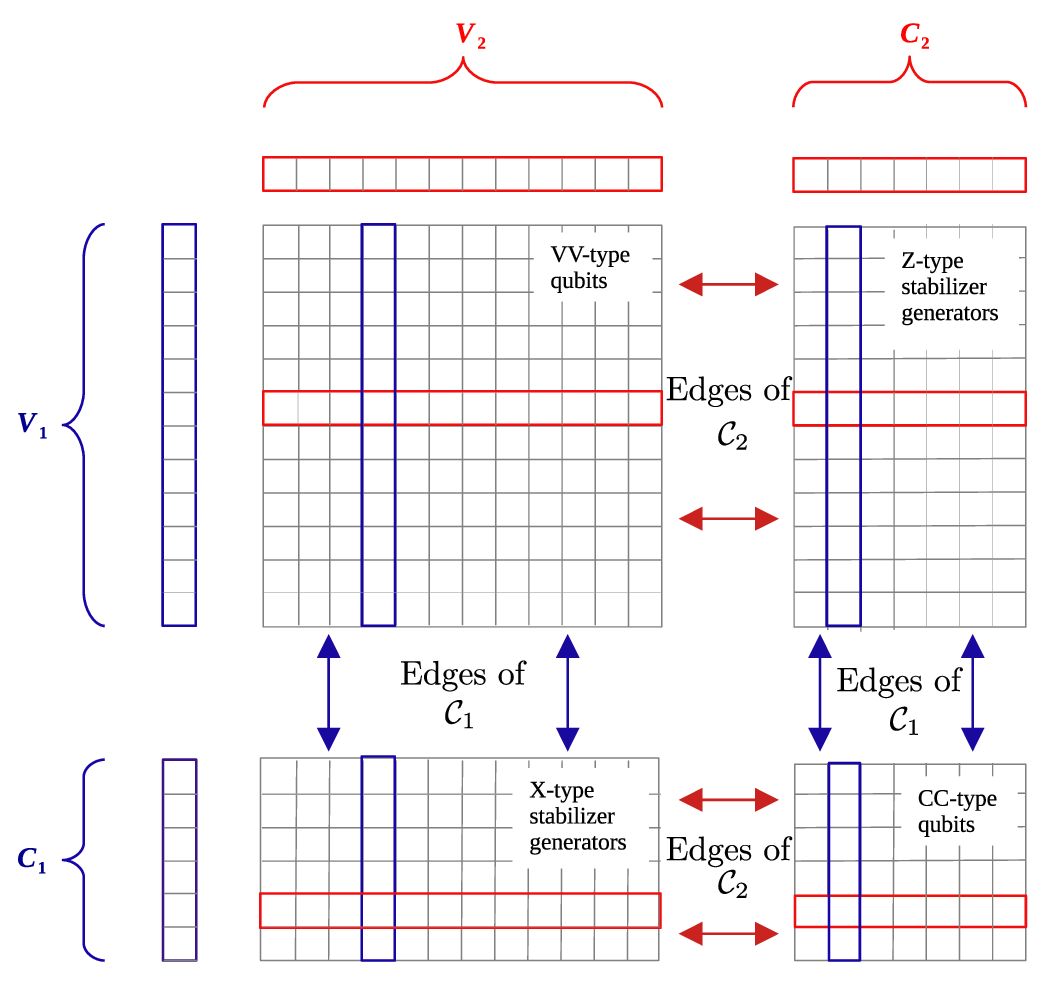}}
   \caption{The Tanner graph for the HGP code constructed from classical codes 
$\mathcal{C}_1$ with Tanner graph $\mathcal{T}(V_1 \cup C_1, E_1)$ 
and $\mathcal{C}_2$ with Tanner graph $\mathcal{T}(V_2 \cup C_2, E_2$).}
    \label{fig:HGP}
\end{figure}

\subsection{Quantum Expander Codes}
\label{subsec:quantum_expander}
The HGP construction of Tillich and Zémor, applied to classical expander codes, results in a class of quantum LDPC codes with constant rate and minimum distance of order \(\Theta(\sqrt{n})\), where $n$ is the length of the resulting quantum code. Such codes are referred to as quantum expander codes. 

Let \(\mathcal{T}(V \cup C, E)\) be a biregular \((\gamma_V, \delta_V, \gamma_C, \delta_C)\)-left-right-expander graph with left degree \(d_V\) and right degree \(d_C\), such that \(d_V \leq d_C\). Consider the classical expander code, \(\mathcal{C}\), associated with the Tanner graph \(\mathcal{T}(H)\), with minimum distance \(d_{\mathrm{min}}\) and parity-check matrix \(H\) of size \(|C| \times |V|\).  

The hypergraph product of \(\mathcal{C}\) with itself results in a quantum expander code with parity-check matrices:  
\[
H_X = \begin{pmatrix}
    I_{|V|} \otimes H~H^T \otimes I_{|C|}
\end{pmatrix},
H_Z = \begin{pmatrix}
    H \otimes I_{|V|}~I_{|C|} \otimes H^T
\end{pmatrix}.
\]
The parameters of this code are
\begin{equation*}
    [[N=|V|^2 + |C|^2,\ k \geq (|V| - |C|)^2,\ \ge \min\{d_{\mathrm{min}}, d_{\mathrm{min}}^T\}]],
\end{equation*}
where \(d_{\mathrm{min}}^T\) is the minimum distance of the transpose code \(\mathcal{C}^T\).
Further, this quantum code is LDPC with generators of weight \((d_V + d_C)\), and each qubit is checked by at most \((2 \max(d_V, d_C))\) generators.

We denote by $V^2$ (resp. $C^2$), the $V \times V$ type (resp. $C \times C$ type) qubits. The $X$-stabilizer generators are indexed by the set $C_X := C \times V$ of size $R_X$, and $Z$-stabilizer generators are indexed by the set $C_Z := V \times C$ of size $R_Z$. Denote by $\mathcal{T}(H_Z)$ (resp. $\mathcal{T}(H_X$) the Tanner graph with the set of variable nodes $V^2 \cup C^2$ and the set of check nodes $C_Z$ (resp. $C_X$).


\subsection{Pruned Peeling Decoder}
\label{subsec:pruned_peeling}
The pruned peeling erasure decoder for quantum CSS codes proposed by Connolly et al. in \cite{connolly2024fast} is inspired by the classical peeling decoder described in  \cite{luby2001efficient},\cite{richardson2001efficient}. 
Consider a classical LDPC code with Tanner graph $\mathcal{T}(H)$. Given an erasure vector $\varepsilon$, two terms are defined: \textit{dangling check} and \textit{dangling bit}. A dangling check is a check node that is incident to a single erased bit and this erased bit is referred to as a dangling bit.

The classical peeling decoder for classical LDPC codes is an iterative algorithm that uses the syndrome value at the  dangling check to determine the erased value at the corresponding dangling bit and then modifies the syndrome accordingly. Using this updated syndrome, the process repeats until the erasures are fully corrected or peeling is not possible further.

A stopping set for $\mathcal{T}(H)$ is defined as a subset of bits that does not contain a dangling bit. Thus, if a stopping set lies within the erasure location, it cannot be corrected using the classical peeling decoder, resulting in decoder failure.
 
 Now, consider decoding the $X$-type errors using the syndrome obtained by measuring $Z$-type stabilizer generators. This problem can be cast as a classical decoding problem of a linear code with Tanner graph $\mathcal{T}(H_Z)$, where $H_Z$ is the parity-check matrix whose rows correspond to $Z$-type stabilizer generators. The classical peeling decoder can then be applied to a quantum CSS code directly; however, it tends to fail more frequently due to an increased number of stopping sets. This increase occurs because the support of an $X$-type generator forms a stopping set for the Tanner graph $\mathcal{T}(H_Z)$. This is due to the fact that the symplectic representation of the $X$-type stabilizers are codewords of the linear code $\ker(H_Z)$. Such stopping sets, induced by stabilizers, are referred to as stabilizer stopping sets.

The idea of the pruned peeling decoder is to iteratively correct all dangling bits using the syndrome at the dangling check. For the remaining errors that are not corrected (as they are part of a stabilizer stopping set), the approach is to identify a stabilizer generator $S$ supported entirely within the erasure locations and then arbitrarily remove a qubit from this support. This works because the error $E$ and the error $ES$ act identically on this qubit. By breaking the stopping sets in this manner, the decoder can correct more erasures. 

The pruned peeling decoder can be applied to any quantum CSS code; however, when applied to hypergraph product codes, it does not show significant improvement over the classical peeling decoder due to the nature of hypergraph stopping sets, as explained next.

\subsection{Stopping Sets for Hypergraph Product Codes}
\label{subsec:stopping_sets_hgp}

Consider a hypergraph product code $\mathrm{HGP}(\mathcal{C}, \mathcal{C})$ following the notations from  Section \ref{subsec:quantum_expander}. Let $\mathcal{T}(H)$ be the Tanner graph of $\mathcal{C}$, and $\mathcal{T}(H^T)$ be the Tanner graph of the transpose code $\mathcal{C}^T$. There are two types of stopping sets for the Tanner graph $\mathcal{T}(H_Z)$:

1. \textit{Stabilizer stopping sets}: These types of stopping sets are induced by the $X$-type stabilizer generators, as explained previously.  
\begin{figure}
    \centering
   \scalebox{0.65} {\includegraphics[width=0.75\linewidth]{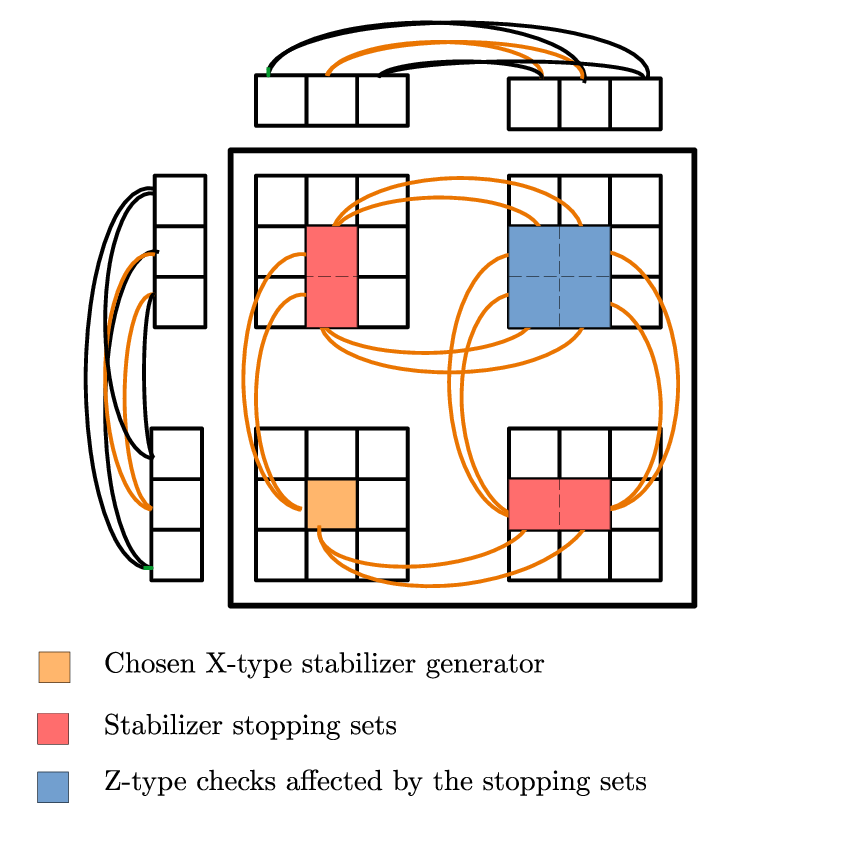}}
    \caption{Stabilizer stopping set for the HGP code constructed from two classical $3$-bit repetition code.}
    \label{fig:stabss}
\end{figure}
See Fig. \ref{fig:stabss} for an example.

2. \textit{Vertical and horizontal stopping sets}: These types of stopping sets arise from the stopping sets of the Tanner graphs $\mathcal{T}(H)$ and $\mathcal{T}(H^T)$. If $S_V$ is a stopping set for $\mathcal{T}(H)$, then for all $v \in V$, $\{v\} \times S_V$ is a stopping set for the Tanner graph $\mathcal{T}(H_Z)$. These stopping sets $\{v\} \times S_V$ are referred to as horizontal stopping sets. Similarly, if $S_C$ is a stopping set for $\mathcal{T}(H^T)$, then for all $c \in C$, $S_C \times \{c\}$ is a stopping set for the Tanner graph $\mathcal{T}(H_Z)$. These stopping sets $S_C \times \{c\}$ are referred to as vertical stopping sets. 
\begin{figure}
    \centering
  \scalebox{0.6}  {\includegraphics[width=1\linewidth]{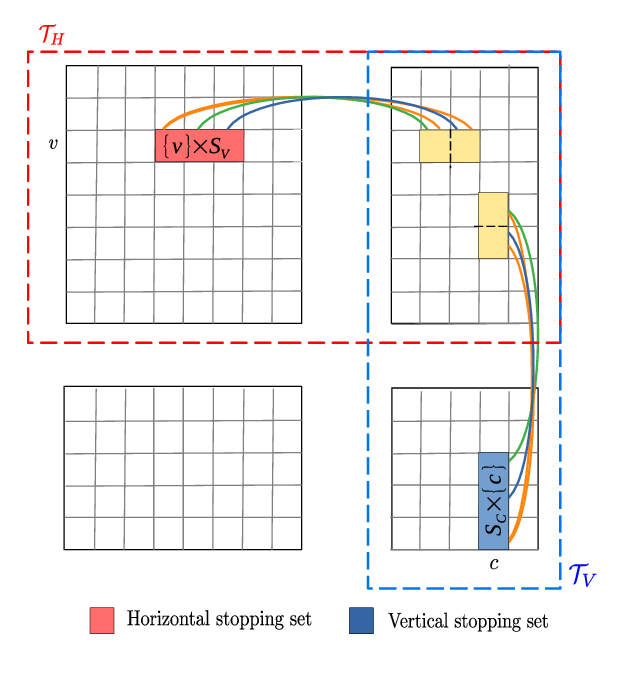}}
    \caption{Horizontal and vertical stopping sets for the HGP code.}
    \label{fig:vhss}
\end{figure}
Fig. \ref{fig:vhss} illustrates this.

The pruned peeling decoder, when applied to hypergraph product codes, can successfully break out of some stabilizer stopping sets. However, this decoder fails when the erasure pattern lies within vertical or horizontal stopping sets. Furthermore, each horizontal (resp. vertical) stopping set in $\mathcal{T}(H)$ (resp. $\mathcal{T}(H^T)$) results in $\sqrt{N}$ copies of the corresponding stopping set in $\mathcal{T}(H_Z)$. Consequently, vertical and horizontal stopping sets account for the majority of failures of the pruned peeling decoder. As a result, pruned peeling decoder provides only a slight improvement over the classical peeling decoder for hypergraph product codes. 

The vertical-horizontal (VH) decoder, discussed next, effectively addresses the vertical and horizontal stopping sets.

\subsection{Vertical-Horizontal (VH) Decoder for Hypergraph Product Codes}
\label{subsec:VH}

After observing that the pruned peeling decoder fails due to horizontal and vertical stopping sets, Connolly et al. \cite{connolly2024fast} proposed the VH erasure decoder for hypergraph product codes. This decoder leverages the product structure of hypergraph product codes to decompose the erasures remaining after pruned peeling into vertical and horizontal clusters, which are then solved sequentially using classical Gaussian decoder.

For the description of the VH erasure decoder, consider the hypergraph product of the classical code $\mathcal{C}=\ker(H)$, having a Tanner graph $(V \cup C,E)$, with itself. The Tanner graph of $\mathrm{HGP}(\mathcal{C},\mathcal{C})$ is given by $\mathcal{T}(V_Q \cup (C_X \cup C_Z), E)$. Let $\mathcal{T}_V$  (resp. $\mathcal{T}_H$) be the subgraph induced by nodes $((V \cup C)\times C)$ (resp. $(V \times (V \cup C))$). The subgraph $\mathcal{T}_V$  (resp. $\mathcal{T}_H$) contains the vertical (resp. horizontal) edges of $\mathcal{T}(V_Q \cup (C_X \cup C_Z), E)$.

Given an erasure vector $\varepsilon$, let $V(\varepsilon)$ be the set of erased qubits and the checks connected to these qubits. A \textit{vertical cluster} (resp. \textit{horizontal cluster}) contains a subset of $V(\varepsilon)$  that forms a connected component in the graph $\mathcal{T}_V$ (resp. $\mathcal{T}_H$).

From the structure of $\mathcal{T}(V_Q \cup (C_X \cup C_Z), E)$, no two vertical clusters (resp. horizontal clusters) can intersect with each other. A vertical cluster and a horizontal cluster can intersect at most in one check, which is referred to as a \textit{connecting check}. A check internal to a single cluster that is not a connecting check is called an \textit{internal check}.

A VH graph for \(\varepsilon\) is a bipartite graph with vertices representing the horizontal and vertical clusters. An edge exists between two vertices in this graph if the corresponding clusters intersect at a connecting check.

A cluster that does not contain any connecting checks, i.e., it only has internal checks, is called an \textit{isolated cluster}. These are the standalone vertices of the VH graph that do not have any edges.
A cluster that contains exactly one connecting check, along with other internal checks, is called a \textit{dangling cluster}. These correspond to the vertices of the VH graph that are connected to exactly one other vertex. A cluster containing more than one connecting check, in other words, one that is neither isolated nor dangling, is called a \textit{non-dangling cluster}. A dangling cluster is further classified as either \textit{free} or \textit{frozen} according the following definition \cite{connolly2024fast}.

\begin{definition}
If there exists an error pattern  within a dangling cluster that has a trivial syndrome at the internal checks and syndrome $1$ at the connecting check, then it is a free check. Otherwise, the check is called frozen.
\end{definition}

Equivalently, a dangling check is \textit{frozen} if, and only if, for all error patterns within a dangling cluster consistent with the original syndrome at the internal checks, the syndrome at the connecting check is either always $0$ or always $1$. This is because if there were two errors, both consistent with the syndrome at the internal checks, but having different values at the connecting check, then the syndrome corresponding to the sum of these two errors is $0$ on all the internal checks and $1$ at the connecting check, implying that the connecting check is free. When the check is classified as free or frozen, we call the corresponding dangling cluster as free or frozen as well. 

The VH decoder begins by classifying each cluster as isolated, dangling or non-dangling. Then, the dangling clusters are classified as free or frozen. An isolated cluster is corrected independently of other clusters, as it is not connected to any other cluster. Gaussian elimination is employed to correct such clusters. For a frozen dangling cluster, although it has a connecting check, this connecting check is frozen, meaning that any correction has the same effect on the connecting check. Therefore, a frozen dangling cluster is corrected in the same manner as an isolated cluster. Correcting a free dangling cluster is postponed until the end of the procedure. During the intermediate steps, this cluster, along with its connecting check, is removed from the Tanner graph $\mathcal{T}(H_Z)$ while the remaining erasures are corrected. After correcting the other erasures and updating the syndrome, a free dangling cluster is corrected using the updated syndrome within that cluster.

This is an iterative algorithm that iterates over all the clusters sequentially. It corrects isolated clusters and frozen dangling clusters independently of other clusters. Once corrected, these clusters are removed from the VH graph, potentially making the remaining clusters correctable.

The Gaussian decoder is used to determine whether a connecting check is frozen or free, as well as for correction within a cluster. The computational complexity of this process is \(O(n^3)\) for a cluster of size \(O(n)\).

This algorithm fails if the VH graph contains a cluster-cycle arising from stabilizer stopping sets. To address this, the authors modify the algorithm by removing free checks from all clusters, not just from free dangling clusters. This modification helps eliminate some cycles but introduces additional complexity, as it requires classifying multiple checks per cluster as frozen or free, rather than just a single check as in the case of dangling clusters. Consequently, the modified algorithm is slower than the VH decoder and has complexity exceeding \(O(n^3)\) for a cluster of size \(O(n)\). However, this modification still does not account for all the stabilizer stopping sets. Therefore, despite this modification, the algorithm fails if a cycle still persists in the VH graph.

\section{A Linear-time erasure decoder for quantum expander codes}
\label{sec:main_result}

Our algorithm can be divided into two phases. In the first phase, we apply the classical peeling decoder to the Tanner graph $\mathcal{T}(H_Z)$. That is, we iteratively correct all dangling bits using the syndrome value at the corresponding dangling checks. The second phase, as given in Algorithm \ref{algo:main_algo}, begins by decomposing the remaining erasures into clusters, as defined and described in Section \ref{subsec:VH}. 
\begin{figure}
    \centering
   \scalebox{0.5} {\includegraphics[width=1\linewidth]{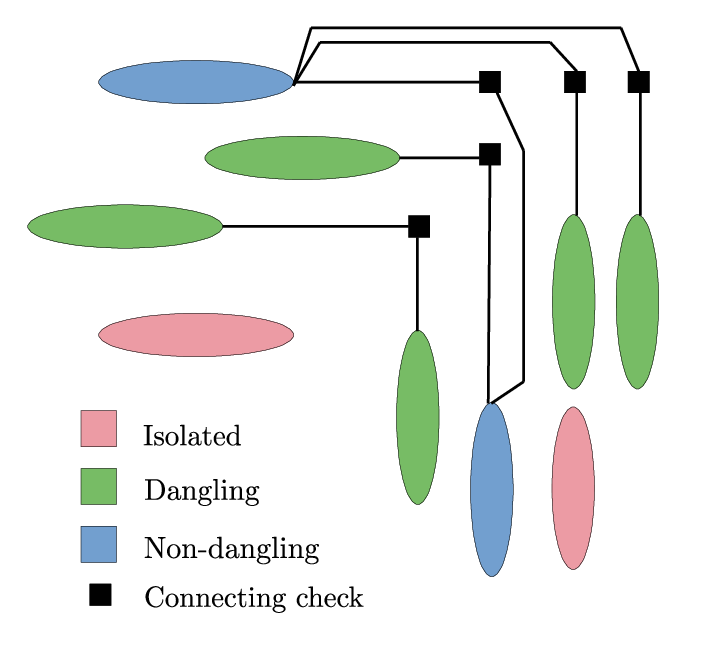}}
    \caption{Illustration of different types of clusters.}
    \label{fig:clusters}
\end{figure}

Fig. \ref{fig:clusters} illustrates this.

Given a cluster $\kappa$, if it is dangling, we first classify it as either free or frozen. If the cluster cannot be classified at this stage, it will remain untouched for now and will be revisited at a later point. The classification procedure is described in detail in Section \ref{subsec:classify_analysis}.

If the cluster $\kappa$ is isolated or frozen dangling, we correct it  using the classical linear-time erasure decoding algorithm \cite{viderman2013linear}, given in Algorithm \ref{algo:peeling}. Classical erasure decoding can be used here, as each horizontal (resp. vertical) isolated or frozen dangling cluster belongs to only a single classical Tanner graph, $\mathcal{T}(H)$ (resp. $\mathcal{T}(H^T)$), corresponding to the classical expander code. This is because such clusters do not have a connecting check in the case of an isolated cluster, or have a frozen connecting check in the case of a frozen dangling cluster.

The notation used in Algorithm \ref{algo:peeling} is as follows: $V_\kappa$ denotes the set of variable nodes and $C_\kappa$ denotes the set of check nodes within the cluster $\kappa$, respectively. For a variable node, $v \in V_\kappa$, $\Gamma(v)$ is a vector of length $R_Z$ supported on the neighborhood of $v$ and $e_v$ is the weight-one binary vector of length $N$ supported on the $v^{\text{th}}$ variable node. Also, $\sigma_{\kappa}$ denotes the syndrome vector of length $R_Z$ restricted to the check nodes in the cluster, and $\sigma_{\kappa, u}$ indicates the syndrome value at the $u^{\text{th}}$ check node.

Algorithm \ref{algo:peeling} begins by initializing a set $\mathrm{Unique}$, which contains all the check nodes in $C_\kappa$ that are connected to only one variable node $v \in V_\kappa$. We will refer to such checks as unique checks. While there exists a variable node $v$ that is connected to a unique check $u$, we check the syndrome value, $\sigma_{\kappa, u}$, at this particular check. If $\sigma_{\kappa, u} = 1$, then we add an error vector supported at the $v^{\text{th}}$ node to the estimate $\hat{E}_i$ and update the syndrome accordingly. Then we remove the variable node $v$ from $V_\kappa$. Next, we check if any of the check nodes in the neighborhood of $v$ have become unique checks. If they have become unique checks, they are added to the set $\mathrm{Unique}$. Then the process repeats.

\begin{algorithm}
    \caption{Algorithm to decode from erasures \cite{viderman2013linear} }
\label{algo:peeling}
\vspace{0.3em}
\textbf{Input:} For a dangling cluster $\kappa$, set of erasure locations $V_\kappa$, set of checks $C_\kappa$, syndrome  $\sigma_{\kappa} \in \mathbb{F}_{2}^{R_Z}$ \\[-0.3em]
\textbf{Output:} Estimate of error $\hat{E}_{\kappa}$ 
\vspace{0.3em}
\hrule
\vspace{0.3em}
\begin{algorithmic} [1]
\Procedure{Peel}{}
\State $\hat{E}_{0}=\mathbf{0}_{N};~ \sigma_{0} = \sigma_{\kappa}  ; ~ i = 0;~ \mathrm{unpeeled} = |V_{\kappa}|$
\State $\mathrm{Unique}=\{u \in  C_\kappa\ :~ |V_\kappa \cap \mathrm{supp}(\Gamma(u))|=1\}$
\While{$\exists ~ v \in V_{\kappa}  ~\text{and } u \in \mathrm{Unique} ~\text{such that}  ~\mathrm{supp}(\Gamma(u)) \cap V_\kappa=v$}
    \If{$\sigma_{\kappa,u}=1$}
        \State $\hat{E}_{i+1}=\hat{E}_i \oplus e_v $
        \State $\sigma_{i+1} = \sigma_{i} \oplus \Gamma(v)$ 
    \EndIf 
    \State $V_\kappa=V_\kappa \backslash \{v\} $
    \For{all $u' \in \mathrm{supp}(\Gamma(v))$}
        \If{$|\mathrm{supp}(\Gamma(u')) \cap V_\kappa|=1$}
        \State $\mathrm{Unique}=\mathrm{Unique} \cup \{u'\}$
         \Else 
         \State $\mathrm{Unique}=\mathrm{Unique} \backslash \{u'\}$
        \EndIf
         
\EndFor
\State $i = i + 1, \mathrm{unpeeled} = \mathrm{unpeeled} - 1 $
\EndWhile
\State \Return $\hat{E}_{i}, \sigma_{i}, \mathrm{unpeeled}$
\EndProcedure
\end{algorithmic}
\end{algorithm}

The estimated error within this cluster, $\hE_\mathrm{PEEL}$, is then added to  $\hat{E}$ , and  $\sigma_\mathrm{PEEL} \oplus \sigma_{\kappa}$  is added to  $\sigma$ . Finally, $V_{\kappa}$ is removed from erasure $\varepsilon $. 

In the case of a free dangling cluster, $\kappa$, the cluster and its free check are removed from the Tanner graph $\mathcal{T}(H_Z)$, added to a stack, and the correction is delayed until all correctable frozen dangling and isolated clusters are corrected. Note that each such cluster becomes essentially an isolated cluster later, so we use Algorithm \ref{algo:peeling} to correct it.

The process of correcting clusters independently can be viewed as peeling each cluster from the VH graph. After iteratively correcting these clusters, the remaining errors will be the ones which couldn't be peeled during any step in the algorithm and those that are part of cluster-cycles in the VH graph arising from the stabilizer stopping sets. These remaining erasures are then corrected using the SSF erasure-decoding algorithm, as described in Algorithm \ref{algo:SSF}.

The SSF algorithm, proposed by Leverrier et al. in \cite{leverrier2015quantum}, is the quantum analogue of the bit-flip error-decoding algorithm introduced by Gallager for classical LDPC codes in \cite{gallager1962low} and analyzed by Sipser and Spielman for classical expander codes in \cite{sipser1996expander}.


The SSF decoder works similar to the bit-flip decoder by dividing its execution into several rounds, where the algorithm attempts to reduce the syndrome weight in each round. However, unlike the bit-flip decoder, where the syndrome weight can be reduced by flipping a single bit, the SSF decoder requires flipping a set of qubits $F$ in order to reduce the syndrome weight. Here, by flipping a qubit, we mean that the qubit undergoes an $X$-type error.

Consider correcting $X$-type errors using the syndrome obtained by measuring $Z$-type stabilizer generators. The authors in \cite{leverrier2015quantum} define a term \textit{small set}, which is a subset of an $X$-type generator's support, and the set of all small sets is denoted by $\mathcal{F}$. The set $F$, which when flipped reduces the syndrome weight, is then chosen from $\mathcal{F}$. We refer the readers to Chapter 4 of \cite{grospellier2019constant} for more details.

SSF for errors can be modified to work for erasures or, more specifically, errors with known locations by defining small sets the following way. A set $F\subseteq V_Q$ is a small set if and only if it is included in the support of an $X$-type generator as well as belongs to the erasure set.  SSF is analysed for erasures in Section \ref{subsec:SSF_analysis}.

Finally, Algorithm \ref{algo:main_algo} returns a Failure if the residual syndrome is not zero; otherwise, it returns the estimated error $\hat{E}$. The analysis of our algorithm is provided in Section \ref{sec:main_analysis}.

\begin{algorithm}
    \caption{Dangling Cluster Classification }
\label{algo:classify}
\vspace{0.3em}
\textbf{Input:} For a dangling cluster $\kappa$, set of erasure locations $V_\kappa$, set of checks $C_\kappa$ \\[-0.3em]
\textbf{Output:} Either $\mathrm{Free}$ or $\mathrm{Unclassified}$ 
\vspace{0.3em}
\hrule
\vspace{0.3em}
\begin{algorithmic} [1]
\Procedure{Classify}{}
\State Set the syndrome ($\sigma_{\kappa}^{\mathrm{new}}$) at the internal checks in $C_{\kappa}$ as $0$ and at the connecting check as $1$
\State $\hE_{\mathrm{PEEL}},~ \sigma_{\mathrm{PEEL}},~ \mathrm{unpeeled}$ = \textsc{PEEL($V_{\kappa}, \sigma_{\kappa}^{\mathrm{new}}$)} 
\If{$\mathrm{unpeeled} \neq 0$}
\State \Return $\mathrm{Unclassified}$ \Comment{Procedure \textsc{PEEL} did not terminate successfully}
\EndIf

\If{ $\sigma_{{\mathrm{PEEL}}_{|C_{\kappa}}} = \mathbf{0}$} \Comment{If the syndrome restricted to the check $C_{\kappa}$ is trivial}
\State \Return $\mathrm{Free}$
\EndIf
\State \Return $\mathrm{Frozen}$
\EndProcedure
\end{algorithmic}
\end{algorithm}

\begin{algorithm}
    \caption{Small-Set-Flip (SSF) Algorithm \cite{grospellier2019constant} }
\label{algo:SSF}
\vspace{0.3em}
\textbf{Input:} An erasure vector $\varepsilon \in \mathbb{F}_{2}^{N}$, a syndrome $\sigma \in \mathbb{F}_{2}^{R_Z}$ \\[-0.3em]
\textbf{Output:} Estimate of error $\hat{E}$ 
\vspace{0.3em}
\hrule
\vspace{0.3em}
\begin{algorithmic} [1]
\Procedure{Small-Set-Flip}{}
\State $\hat{E}_{0}=\mathbf{0}_{N};~ \sigma_{0} = \sigma_{\kappa}  ; ~ i = 0$
\While{$\exists F$ such that $\mathrm{supp}(F)  \in \mathcal{F}: |\sigma_i|-|\sigma_i \oplus \sigma(F)|\geq \beta d_C |\mathrm{supp}(F)|$}
\State $ \mathrm{argmax}_{F \in \mathcal{F}} \frac{|\sigma_i| - |\sigma_i \oplus \sigma(F)|}{|F|}$
\State $\hE_{i+1} = \hE_{i} \oplus F_{i}$
\State $\sigma_{i+1} = \sigma_{i} \oplus \sigma(F_{i})$
\State $i = i + 1 $
\EndWhile
\State \Return $\hE_{i}, \sigma_{i}$
\EndProcedure
\end{algorithmic}
\end{algorithm}

\begin{algorithm}
    \caption{Linear-time Erasure Decoder for Quantum Expander Codes }
\label{algo:main_algo}
\vspace{0.3em}
\textbf{Input:} An erasure vector $\varepsilon \in \mathbb{F}_{2}^{N}$, a syndrome $\sigma \in \mathbb{F}_{2}^{R_Z}$ \\[-0.3em]
\textbf{Output:} Either Failure or an $X$-type error $\hat{E} \in \{I,X\}^{N}$ such that $\sigma(\hat{E})=\sigma$ and $\mathrm{supp}(\hat{E}) \subseteq \mathrm{supp}(\varepsilon)$
\vspace{0.3em}
\hrule
\vspace{0.3em}
\begin{algorithmic} [1]
\Procedure{Linear-time Erasure Decoder}{}
\State $\hE = \mathbf{0}_{N}$ \Comment{$\hE$: error estimate, initialized as a zero vector}
\State $E_{\mathrm{res}} = \varepsilon$ \Comment{$E_{\mathrm{res}}$: set of residual erasures}
\State $L$ = [] \Comment{$L$: an empty stack}
\While{there exists an isolated or a dangling cluster $\kappa$}
\If{$\kappa$ is dangling}
\State $\mathrm{output} =$ \textsc{CLASSIFY($V_{\kappa}'$)} \Comment{$V_{\kappa}'$: a copy of $V_{\kappa}$}
\If{$\mathrm{output}= \mathrm{Unclassified}$}
\State \textbf{continue} \Comment{Go to line 5}
\EndIf
\EndIf
\If{$\kappa$ is $\mathrm{Isolated}$ or $\mathrm{Frozen}$}
\State $\hE_{\mathrm{PEEL}}, \sigma_{\mathrm{PEEL}}$ = \textsc{PEEL($V_{\kappa}, \sigma_{\kappa}$)} 
\State $\hE = \hE \oplus \hE_{\mathrm{PEEL}},~ E_{\mathrm{res}} = E_{\mathrm{res}} \oplus \hE_{\mathrm{PEEL}},~\sigma = \sigma \oplus (\sigma_{\mathrm{PEEL}} \oplus \sigma_{\kappa})$
\Else \Comment{When $\kappa$ (and the connecting check) is free}
\State Remove the free connecting check $c$ of $\kappa$ from the Tanner Graph $\mathcal{T}(H_{Z})$
\State Add the pair $(\kappa,c)$ to the stack $L$
\State Set $E_{\mathrm{res}_{|V_{\kappa}}} = \mathbf{0}_{|V_{\kappa}|}$ \Comment{$E_{\mathrm{res}_{|V_{\kappa}}}$: residual errors restricted to $V_{\kappa}$}
\EndIf
\EndWhile
\While{the stack $L$ is non-empty}
\State Pop a pair $(\kappa,c)$ from the stack $L$
\State Add the check node $c$ to the Tanner Graph $\mathcal{T}(H_{Z})$
\State $\hE_{\mathrm{PEEL}}, \sigma_{\mathrm{PEEL}}$ = \textsc{PEEL($V_{\kappa}, \sigma_{\kappa}$)}
\State $\hE = \hE \oplus \hE_{\mathrm{PEEL}},~ E_{\mathrm{res}} = E_{\mathrm{res}} \oplus \hE_{\mathrm{PEEL}},~\sigma = \sigma \oplus (\sigma_{\mathrm{PEEL}} \oplus \sigma_{\kappa})$
\EndWhile
\State $\hE_{\mathrm{SSF}}, \sigma_{\mathrm{SSF}}$ = \textsc{SMALL-SET-FLIP($E_\mathrm{res}, \sigma$)} 
\State $\hE = \hE \oplus \hE_{\mathrm{SSF}},~ E_\mathrm{res} = E_\mathrm{res} \oplus \hE_{\mathrm{SSF}}$
\If{$\sigma_{\mathrm{SSF}} \neq \mathbf{0}_{R_Z}$}
\State \Return Failure
\EndIf
\State \Return $\hE$

\EndProcedure
\end{algorithmic}
\end{algorithm}

\begin{remark}
    Algorithm \ref{algo:main_algo} (in the current form) may result in an infinite loop if the condition on line 8 is always true after some point, due to classification failure. To avoid this, we can keep a variable $\hE_{i}$ for the $i^{\text{th}}$ iteration and check if the error in the current iteration is the same as in the previous one. If it is, then we break the while loop. For simplicity, we have not added this in the algorithm.
\end{remark}

\section{Analysis of Algorithm \ref{algo:main_algo}}
\label{sec:main_analysis}

In this section, we provide the proof of correctness of Algorithm \ref{algo:main_algo}.

\subsection{Analysis of Algorithm \ref{algo:peeling}} \label{subsec:peeling_analysis}
\begin{theorem} \label{theo:peeling}
    Consider a classical expander code associated with a bipartite Tanner graph $\mathcal{T}(V \cup C, E)$, which is $(\gamma, \delta)$-left expanding for constants $\gamma, \delta > 0$. Any error $E \subseteq \varepsilon$, where $\varepsilon$ is the set of erasure locations such that $|\varepsilon| \leq \gamma |V|$, can be corrected by Algorithm \ref{algo:peeling} in linear time.
\end{theorem}
\begin{proof}
    We refer the readers to Section 4.2 in Ref.~\cite{viderman2013linear} for the proof.

\end{proof}
In Algorithm \ref{algo:main_algo}, we use Algorithm \ref{algo:peeling} to correct an isolated or frozen dangling cluster $\kappa$. If $|V_\kappa| \leq \Gamma_V |V|$ (resp. $|V_\kappa| \leq \Gamma_C |C|$) for a horizontal (resp. vertical) cluster, then the erasures in this cluster are corrected completely according to Theorem \ref{theo:peeling}.

\subsection{Analysis of Algorithm \ref{algo:classify}}
\label{subsec:classify_analysis}



After applying the classical peeling decoder to the Tanner graph $\mathcal{T}(H_Z)$, none of the internal checks for any cluster would be dangling. It follows that if there is a dangling check in a cluster, then it has to be the connecting check. Thus, peeling for any dangling cluster either for classification or for erasure-correction must begin with the connecting check. 

For classification, we fix the (virtual) syndrome at the internal checks as $0$ and at the connecting check as $1$ and use $\mathrm{PEEL}$ to assign values to variable nodes consistent with this syndrome. If the procedure does not terminate, some of the variables nodes are left unassigned. In this case, the $\mathrm{PEEL}$ procedure has failed to classify.

Since the procedure $\mathrm{PEEL}$ runs in linear-time, classification is also done in linear-time.


\subsection{Analysis of Algorithm \ref{algo:SSF}}
\label{subsec:SSF_analysis}

The notation followed in this analysis is: Let $\mathcal{T}(V \cup C, E)$ be a $(d_V,d_C)$-biregular $(\gamma_V,\delta_V,\gamma_C,\delta_C)$-left-right-expanding graph. Let $\mathcal{C}$ be the classical expander code associated with $\mathcal{T}$. The quantum expander code $\mathcal{C}_{\mathcal{T}} = \mathrm{HGP}(\mathcal{C}, \mathcal{C})$ is obtained by taking the hypergraph product of $\mathcal{C}$ with itself.

Define $r$ and $\beta$ by:
\begin{equation*}
    r := \frac{d_V}{d_C} = \frac{|C|}{|V|}, \quad \beta := \frac{r}{2}\left[1 - 4\left(\delta_V + \delta_C + (\delta_C - \delta_V)^2\right)\right].
\end{equation*}
The decoding strategy of SSF, given in Algorithm \ref{algo:SSF}, is to decrease the syndrome by flipping a set of qubits called a small set. Let the set of all small sets be given by 
\[
\mathcal{F} := \{F \subseteq \Gamma_X(g) \cap \varepsilon : g \in C_X\}.
\] Here, $\Gamma_X(g)$ represents the support of the $X$-type stabilizer generator $g$, and $\varepsilon$ denotes the set of erasures. However, in Algorithm \ref{algo:SSF}, the small set, the set of erasures, and the support of stabilizer generators are treated as vectors supported appropriately.

For correcting erasures, the idea is to go through the set of all small sets $\mathcal{F}$ and check whether there exists a small set $F \in \mathcal{F}$ that would decrease the syndrome. If such an $F$ exists, we flip $F$ and then proceed to find the next small set. In Algorithm \ref{algo:SSF} a small set $F$ is flipped if and only if it decreases the syndrome bt atleast $\beta d_C |F|$. The analysis of the SSF decoder for erasures follows a similar analysis as done for errors in \cite{fawzi2018efficient}. Given the syndrome $\sigma$, Algorithm \ref{algo:SSF} provides an estimate of the error $\hat{E}$, supported on the erasure locations, which satisfies the syndrome.

Following the references \cite{leverrier2015quantum} and \cite{fawzi2018efficient}, we define the notion of critical generators. An $X$-type generator is termed a critical generator if its support contains a set of qubits that can be flipped to reduce the syndrome.

\begin{definition}
    A generator $g \in C_X$ is said to be a critical generator for $E \subseteq \varepsilon$ if:\begin{equation*}
    \Gamma_X(g)=\Gamma_1 \uplus \overline{\Gamma}_1 \uplus \Gamma_2 \uplus \overline{\Gamma}_2
    \end{equation*} where \begin{itemize}
        \item $\Gamma_X(g) \cap V^2= \Gamma_1 \uplus \overline{\Gamma}_1 $ and $\Gamma_X(g) \cap C^2= \Gamma_2 \uplus \overline{\Gamma}_2 $;
        \item $\Gamma_1, \Gamma_2 \subseteq E$ and  $\overline{\Gamma}_1,\overline{\Gamma}_2 \subseteq \varepsilon \setminus E$
        \item for all $v_1 \in \Gamma_1$, $v_2 \in \Gamma_2$, $\overline{v}_1 \in \overline{\Gamma}_1$ and $\overline{v}_2 \in \overline{\Gamma}_2$: \begin{itemize}
            \item $E \cap [\Gamma_Z[\Gamma_Z(v_1)\cap \Gamma_Z(v_2)]=\{v_1,v_2\}$ 
            \item $E \cap [\Gamma_Z[\Gamma_Z(\overline{v}_1)\cap \Gamma_Z(\overline{v}_2)]=\phi$
            \item $E \cap [\Gamma_Z[\Gamma_Z(v_1)\cap \Gamma_Z(\overline{v}_2)]=\{v_1\}$
            \item $E \cap [\Gamma_Z[\Gamma_Z(\overline{v}_1)\cap \Gamma_Z(v_2)]=\{v_2\}$
        \end{itemize}
        \item $\Gamma_1 \cup \Gamma_2 \neq \phi$.
    \end{itemize}
\end{definition}

\begin{lemma}(\cite[Lemma B.2]{fawzi2018efficient}) \label{lemma:critical}
    For an error $E \subseteq V_Q$ such that $0 < |E| \leq \mathrm{min}(\gamma_V |V|, \gamma_C |C|) $, there exists a critical generator for $E$.
\end{lemma}

\begin{lemma} (\cite[Lemma B.3]{fawzi2018efficient}) \label{lemma:syn_decrease}
    Let $E \subseteq V_Q$ be a error such that $0 < |E| \leq r~ \mathrm{min}(\gamma_V |V|, \gamma_C |C|) $ then there exists an error $F \subseteq \mathcal{F}$ with $|\sigma(E)|-|\sigma(E \oplus F )| \geq \beta d_C |F|$.
\end{lemma}

\begin{theorem} \label{theo:ssf}
    Let $\mathcal{T}(V \cup C, E)$ be a $(d_V, d_C)$-biregular $(\gamma_V, \delta_V, \gamma_C, \delta_C)$-left-right-expanding graph such that $\delta_V, \delta_C < 1/8$. For a quantum expander code $\mathcal{C}_{\mathcal{T}}$, the small-set-flip decoder runs in linear time in the code length $n = |V|^2 + |C|^2$ and decodes any erasure pattern $\varepsilon$ where
    \begin{equation*}
        |\varepsilon| \leq r \, \mathrm{min}(\gamma_V |V|, \gamma_C |C|).
    \end{equation*}
\end{theorem}

\begin{proof}

We run Algorithm \ref{algo:SSF} on the input $E_0$, and let $\hE_f$ denote the output. Define 
\[
E_f = E_0 \oplus \hE_f = E_0 \oplus \left(\bigoplus_{k=0}^{f-1} F_k\right).
\]
We also define $U := E_0 \cup F_0 \cup \cdots \cup F_{f-1}$,
referred to as the \textit{execution support}, where $f \in \mathbb{N}$ represents the number of iterations. This set includes all qubits that were in error at some point during the algorithm's execution.

By definition, $F_i \subseteq \varepsilon \ \forall~ i$ and $E_0 \subseteq \varepsilon$, which implies $U \subseteq \varepsilon$. Consequently, $|U| \leq |\varepsilon|$, leading to 
\[
|U| \leq r \, \mathrm{min}(\gamma_V |V|, \gamma_C |C|).
\]
Thus, 
\[
|E_f| = |E_0 \oplus \hE_f| \leq |U| \leq r \, \mathrm{min}(\gamma_V |V|, \gamma_C |C|).
\]

To prove Theorem \ref{theo:ssf}, it is sufficient to show that $E_f = 0$. Assume, for the sake of contradiction, that $E_f \neq 0$. Then, by the hypothesis of Lemma \ref{lemma:syn_decrease}, there exists $F \in \mathcal{F}$ such that $|\sigma(E_f)| - |\sigma(E_f \oplus F)| \geq \beta d_C |F|$.
However, since $E_f$ is the output, the while condition in Algorithm \ref{algo:SSF} is not satisfied for $\sigma_f = \sigma(E_0 \oplus \hE_f)$, leading to a contradiction.
    
\end{proof}

In Algorithm \ref{algo:main_algo}, we use Algorithm \ref{algo:SSF} to correct the errors which were not corrected in the previous iterations of peeling. From Theorem \ref{theo:ssf}, if the number of such errors is up to $\frac{d_v}{d_c}\mathrm{min}(\gamma_V |V|, \gamma_C |C|)$, which is a fraction of the minimum distance of the quantum expander code, then they will be corrected completely by SSF.
\section{Simulation Results} \label{sec:simulation}
In Fig. \ref{fig:failure_rate}, we present simulation results for the performance of the $[[1525,25]]$, $[[3904,64]]$, $[[6100,100]]$, and $[[8784,144]]$ quantum expander codes obtained by hypergraph product of $(5,6)$-biregular classical expander codes under the peeling decoder.  The plots are generated by performing $10^5$ trials for each erasure rate, and with help from \cite{connolly:2022}.
\begin{figure}[ht]
    \centering
    \scalebox{0.85}{\includegraphics[width=0.8\linewidth]{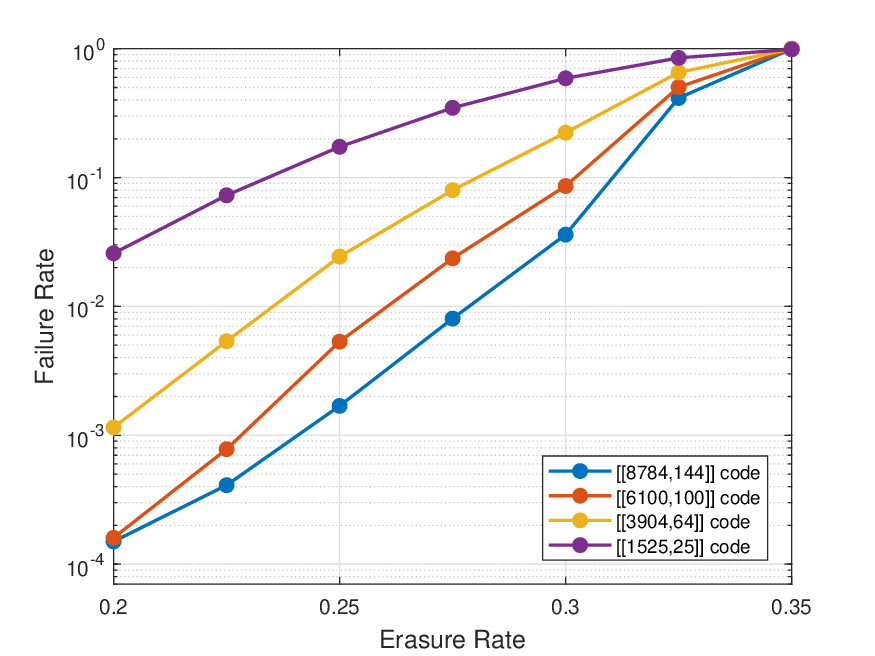}}
    \caption{Performance of the $[[1525,25]]$, $[[3904,64]]$, $[[6100,100]]$, and $[[8784,144]]$ quantum expander codes under the peeling decoder.}
    \label{fig:failure_rate}
    \vspace{-0.25em}
\end{figure}

As Fig. \ref{fig:failure_rate} shows, for the same erasure rate, the failure rate decreases significantly with the length of the quantum expander code, for low and moderate erasure rates.  This is interesting, and in contrast with the simulation results in \cite{connolly2024fast}, where increasing the length of the considered sub-class of HGP codes (progressive-edge-growth-based HGP codes) does not improve the performance of the peeling decoder. For example, for an erasure rate of $0.25$, the failure rate of the HGP codes considered in \cite{connolly2024fast} remains almost constant at around $10^{-1}$. 
Further, for the same erasure rate, the $[[1600,64]]$ HGP code decoded using the VH decoder in \cite{connolly2024fast} can tolerate a failure rate of approximately $7$ in $10^{-3}$, but at the cost of increased complexity. However, as Fig. \ref{fig:failure_rate} shows, a longer quantum expander code, for example $[[6100,100]]$ code, can instead be employed to achieve a comparable performance under peeling, which has linear complexity.

In Table \ref{tab:residual_error}, we tabulate the maximum weight of the residual error after peeling for the $[[1525,25]]$, $[[3904,64]]$, $[[6100,100]]$, and $[[8784,144]]$ quantum expander codes.
In Table \ref{tab:mean_and_var}, we tabulate the mean and variance of the residual error weight after peeling. For example, for the $[[6100,100]]$ code, at erasure rate $0.275$, the maximum weight of a residual error was observed to be 41, while the empirical mean and variance of the weight of residual error patterns were observed to be approximately $0.26$ and $3.12$, respectively. Similarly, for the $[[8784,144]]$ code, at erasure rate $0.3$, the maximum weight of a residual error was observed to be $38$, while the empirical mean and variance of the weight of residual error patterns were observed to be approximately $0.44$ and $6.09$, respectively. This implies that in many cases one can potentially employ the modified SSF algorithm (with known error locations) after peeling to correct the remaining erasures. Since the modified SSF algorithm has linear complexity, it follows that the combination of peeling followed by the modified SSF algorithm also has linear complexity.

In Table \ref{tab:isolated_horizontal_clusters} and \ref{tab:isolated_vertical_clusters}, we tabulate some statistics on the number and sizes of isolated horizontal clusters and isolated vertical clusters after peeling, respectively. We observe empirically that, for a fixed erasure rate, the maximum number of isolated horizontal (and vertical) clusters does not increase as a function of the length of the code. Therefore, these clusters can potentially be corrected using Gaussian Elimination. 


We also observed that the peeling procedure inside clusters was rarely successful because after the initial peeling, the connecting checks in the clusters had degree $>1$ inside the cluster. However, even if the degree of the connecting check is $1$, it does not guarantee successful termination of peeling.


\begin{table}[h!]
\centering
\begin{tabular}{|c|c|c|c|c|c|c|}
\hline
\textbf{Code $\backslash$ Erasure Rate}  & $\mathbf{0.2}$ & $\mathbf{0.225}$ & $\mathbf{0.25}$ & $\mathbf{0.275}$ & $\mathbf{0.3}$ & $\mathbf{0.325}$  \\
\hline
$[[1525,25]]$ & $20$ & $19$ & $30$ & $42$ & $250$ & $289$ \\
\hline
$[[3904,64]]$ & $15$ & $16$ & $33$ & $38$ & $534$ & $663$ \\
\hline
$[[6100,100]]$ & $15$ & $20$ & $21$ & $41$ & $738$ & $994$   \\
\hline
$[[8784,144]]$ & $13$ & $21$ & $21$ & $24$ & $38$ & $1378$ \\
\hline
 \end{tabular}
\caption{Maximum weights of residual error after peeling for the $[[1525,25]]$, $[[3904,64]]$, $[[6100,100]]$, and $[[8784,144]]$ quantum expander codes.}
\label{tab:residual_error}
\end{table}

\begin{table}[h!]
    \centering
    \resizebox{\textwidth}{!}{
    \begin{tabular}{|c|c|c|c|c|c|c|c|c|}
        \hline
        \textbf{Erasure Rate $\backslash$ Code} & \multicolumn{2}{c|}{\textbf{[[1525,25]]}} & \multicolumn{2}{c|}{\textbf{[[3904,64]]}} & \multicolumn{2}{c|}{\textbf{[[6100,100]]}} & \multicolumn{2}{c|}{\textbf{[[8784,144]]}} \\ \hline
        & \textbf{Mean} & \textbf{Variance} & \textbf{Mean} & \textbf{Variance} & \textbf{Mean} & \textbf{Variance} & \textbf{Mean} & \textbf{Variance} \\  \hline
        0.2 & 0.15 & 0.94 & 0.0094 & 0.085 & 0.001 & 0.011 & 0.00093 & 0.007  \\ \hline
        0.225 & 0.43 & 2.79 & 0.047 & 0.45 & 0.0080 & 0.94 & 0.0031 & 0.03  \\ \hline
        0.25 & 1.12 & 7.33 & 0.22 & 2.11 & 0.056 & 0.65 & 0.014 & 0.16  \\ \hline
        0.275 & 2.59 & 17.61 & 0.76 & 7.51 & 0.26 & 3.12 & 0.09 & 1.08  \\ \hline
        0.3 & 7.29 & 382.95 & 2.68 & 169.70 & 0.94 & 26.26 & 0.44 & 6.09 \\ \hline
        0.325 & 66.78 & 8269.12 & 188.38 & 60230.55 & 272.09 & 140976.6 & 371.67 & 281472.82 \\ \hline
    \end{tabular}}
    \caption{Mean and variance of weights of residual error after peeling for the $[[1525,25]]$, $[[3904,64]]$, $[[6100,100]]$, and $[[8784,144]]$ quantum expander codes.}
    \label{tab:mean_and_var}
\end{table}
\vspace{-1em}

\begin{table}[h!]
    \centering
    \resizebox{0.9\textwidth}{!}{
    \begin{tabular}{|c|c|c|c|c|c|}
        \hline
        &
        {\textbf{Erasure Rate $\backslash$ Code}} & $\mathbf{[[1525, 25]]}$ & $\mathbf{[[3904, 64]]}$ & $\mathbf{[[6100, 100]]}$ & $\mathbf{[[8784, 144]]}$ \\ \hline
        \multirow{6}{*}{\textbf{Maximum number of isolated horizontal clusters}} & 
        $0.2$ & $2$ & $1$ & $1$ & $1$  \\ \cline{2-6}
        & $0.225$ & $3$ & $1$ & $1$ & $1$  \\ \cline{2-6}
        & $0.25$ & $4$ & $3$ & $2$ & $1$  \\ \cline{2-6}
        & $0.275$ & $5$ & $3$ & $3$ & $2$  \\ \cline{2-6}
        & $0.3$ & $7$ & $4$ & $4$ & $3$  \\ \cline{2-6}
        & $0.325$ & $7$ & $6$ & $5$ & $3$  \\ \hline
        \multirow{6}{*}{\textbf{Maximum size of isolated horizontal clusters}} & 
        $0.2$ & $18$ & $24$ & $26$ & $26$  \\ \cline{2-6}
        & $0.225$ & $20$ & $26$ & $30$ & $29$  \\ \cline{2-6}
        & $0.25$ & $22$ & $32$ & $31$ & $37$  \\ \cline{2-6}
        & $0.275$ & $22$ & $33$ & $36$ & $39$  \\ \cline{2-6}
        & $0.3$ & $23$ & $32$ & $36$ & $45$  \\ \cline{2-6}
        & $0.325$ & $23$ & $34$ & $36$ & $44$  \\ \hline
        \multirow{6}{*}{\textbf{Minimum size of isolated horizontal clusters}} &
        $0.2$ & $7$ & $10$ & $0$ & $21$  \\ \cline{2-6}
        & $0.225$ & $7$ & $11$ & $16$ & $25$  \\ \cline{2-6}
        & $0.25$ & $7$ & $10$ & $14$ & $19$  \\ \cline{2-6}
        & $0.275$ & $7$ & $10$ & $13$ & $17$  \\ \cline{2-6}
        & $0.3$ & $7$ & $10$ & $13$ & $14$  \\ \cline{2-6}
        & $0.325$ & $7$ & $10$ & $13$ & $14$  \\ \hline
    \end{tabular}}
    \caption{Statistics on isolated horizontal clusters after peeling for the $[[1525,25]]$, $[[3904,64]]$, $[[6100,100]]$, and $[[8784,144]]$  quantum expander codes.}
    \label{tab:isolated_horizontal_clusters}
\end{table}

\begin{table}[h!]
    \centering
    \resizebox{0.9\textwidth}{!}{
    \begin{tabular}{|c|c|c|c|c|c|}
        \hline
        &
        {\textbf{Erasure Rate $\backslash$ Code}} & $\mathbf{[[1525, 25]]}$ & $\mathbf{[[3904, 64]]}$ & $\mathbf{[[6100, 100]]}$ & $\mathbf{[[8784, 144]]}$  \\ \hline
        \multirow{6}{*}{\textbf{Maximum number of isolated vertical clusters}} & 
        $0.2$ & $1$ & $0$ & $0$ & $0$  \\ \cline{2-6}
        & $0.225$ & $2$ & $1$ & $1$ & $0$  \\ \cline{2-6}
        & $0.25$ & $2$ & $1$ & $1$ & $0$  \\ \cline{2-6}
        & $0.275$ & $2$ & $1$ & $1$ & $0$  \\ \cline{2-6}
        & $0.3$ & $2$ & $2$ & $1$ & $1$  \\ \cline{2-6}
        & $0.325$ & $3$ & $2$ & $1$ & $1$ \\ \hline
        \multirow{6}{*}{\textbf{Maximum size of isolated vertical clusters}} & 
        $0.2$ & $16$ & $0$ & $0$ & $0$  \\ \cline{2-6}
        & $0.225$ & $17$ & $22$ & $24$ & $0$ \\ \cline{2-6}
        & $0.25$ & $18$ & $23$ & $26$ & $0$ \\ \cline{2-6}
        & $0.275$ & $18$ & $24$ & $26$ & $0$  \\ \cline{2-6}
        & $0.3$ & $20$ & $27$ & $32$ & $34$  \\ \cline{2-6}
        & $0.325$ & $20$ & $28$ & $31$ & $36$  \\ \hline
        \multirow{6}{*}{\textbf{Minimum size of isolated vertical clusters}} &
        $0.2$ & $8$ & $0$ & $0$ & $0$  \\ \cline{2-6}
        & $0.225$ & $8$ & $19$ & $0$ & $0$  \\ \cline{2-6}
        & $0.25$ & $8$ & $15$ & $0$ & $0$ \\ \cline{2-6}
        & $0.275$ & $8$ & $14$ & $22$ & $0$  \\ \cline{2-6}
        & $0.3$ & $8$ & $13$ & $17$ & $23$  \\ \cline{2-6}
        & $0.325$ & $8$ & $13$ & $18$ & $26$  \\ \hline
    \end{tabular}}
    \caption{Statistics on isolated vertical clusters after peeling for the $[[1525,25]]$, $[[3904,64]]$, $[[6100,100]]$, and $[[8784,144]]$ quantum expander codes.}
    \label{tab:isolated_vertical_clusters}
\end{table}

\section{Conclusion and Future Work} \label{sec:conclusions} 

 In this paper, it was shown through simulation results that compare across quantum LDPC codes, that an attractive option to be considered, is the use of quantum expander codes in conjunction with the linear-complexity peeling decoder.  As in the classical case, stopping sets pose a major challenge in designing efficient erasure decoders for quantum LDPC codes. The linear-time decoding algorithm presented here employed small-set-flip decoding as well as cluster-based decoding to handle stopping sets.  However, it was our finding through simulation, that the cluster-based decoding improves only in small measure, upon the performance of the peeling decoder. Our simulation results also include limited statistical data compiled on erasure patterns that the peeling decoder was unable to recover from. For future work, an interesting direction is to employ other techniques to handle the small number of residual errors after peeling.




\section*{Acknowledgements}

The work of V. Lalitha is supported in part by the grant DST/INT/RUS/RSF/P-41/2021
from the Department of Science \& Technology, Govt. of India. The work of Shobhit Bhatnagar is supported by a Qualcomm Innovation Fellowship India 2021. The authors thank Anirudh Krishna for useful comments and suggestions.

\newpage

\bibliography{IEEEabrv,main.bib}
\bibliographystyle{ieeetr}

\end{document}